\documentclass[12pt]{article}
\usepackage[intlimits]{amsmath}
\usepackage{amssymb,amsthm,slashed,mathabx}

\usepackage[T1]{fontenc}
\usepackage[cp1250]{inputenc}

\newcommand{\<}{\langle}
\renewcommand{\>}{\rangle}

\newcommand{\F}{\mathcal{F}}

\newcommand{\C}{\mathcal{C}}
\newcommand{\Hc}{\mathcal{H}}
\newcommand{\K}{\mathcal{K}}

\newcommand{\mR}{\mathbb{R}}
\newcommand{\1}{{\bf 1}}
\newcommand{\clc}{\overline{V_+}}

\newcommand{\lp}{\left(}
\newcommand{\rp}{\right)}
\newcommand{\as}{\mathrm{as}}
\newcommand{\out}{\mathrm{out}}
\newcommand{\al}{\alpha}
\newcommand{\ep}{\epsilon}
\newcommand{\la}{\lambda}
\newcommand{\La}{\Lambda}
\newcommand{\dsp}{\displaystyle}

\newcommand{\vep}{\varepsilon}
\newcommand{\vp}{\varphi}
\newcommand{\ti}{\widetilde}
\newcommand{\w}{\omega}
\newcommand{\W}{\Omega}
\newcommand{\ov}{\overline}
\newcommand{\wh}{\widehat}
\newcommand{\wch}{\widecheck}
\newcommand{\p}{\partial}
\newcommand{\con}{\mathrm{const}}
\newcommand{\schi}{\chi^\diamond}

\DeclareMathOperator{\supp}{supp}
\DeclareMathOperator{\id}{id}
\newtheorem{thm}{Theorem}
\newtheorem{pr}[thm]{Proposition}
\newtheorem{lem}[thm]{Lemma}
\newtheorem{col}[thm]{Corollary}
\newtheorem{ass}{Assumption}
\newtheorem{dfn}{Definition}
\title{Infraparticle problem, asymptotic fields\\ and Haag-Ruelle theory}
\author{Andrzej Herdegen\thanks{e-mail: herdegen@th.if.uj.edu.pl}\\
{\it Institute of Physics, Jagiellonian University,}\\
{\it Reymonta 4, 30-059 Kraków, Poland}}
\date{}

\begin{document}

\maketitle

\begin{abstract}
In this article we want to argue that an appropriate generalization of the Wigner concepts may lead to an asymptotic particle with well-defined mass, although no mass hyperboloid in the energy-momentum spectrum exists.

\end{abstract}

\section{Introduction}\label{int}

It is a well-established fact, that electrically charged particles do not produce a discrete mass hyperboloid in the mass spectrum of the quantum theory in which they participate. The evidence comes both from the perturbational QED, as well as from more fundamental arguments \cite{bu86}. Therefore, the Wigner concept of an elementary particle as a carrier of an irreducible unitary representation of (the universal covering group of) the restricted Poincar\'e group does not apply to these particles. Also, the absence of a discrete mass hyperboloid has posed considerable difficulties in obtaining a manageable scattering theory for such ``infraparticles'', as they were named long ago \cite{sc63}. Much effort has been devoted to the construction of asymptotic charged states as vectors in the Hilbert space of various models by ``dressing'' a charged particle with a~``cloud of radiation''; some recent examples include \cite{ch10}, \cite{ch09}, where also further bibliographic information may be found.  Nevertheless, it seems that a~convincing characterization of particle-like charged states in relativistic QFT has not been reached yet.

Another approach developed in \cite{bu91}, \cite{po04}, aims at an alternative, with respect to the Wigner concept, general characterization of a particle. This approach is based on a study, within the algebraic approach to QFT, of the effect of localization of a particle in the process of measurement. The resulting theory of asymptotic functionals and particle weights may be viewed as an extension of the Dirac notion of a particle as an improper eigenstate of the energy-momentum vector. However, there is one essential point of departure from the usual quantum-mechanical notion of improper eigenstates: weights with different four-momentum characteristics do not interfere. The scattering theory must then be based on scattering cross-sections only, rather than amplitudes. Is this indeed the price to be paid for inclusion of charged particles into the scattering theory? We believe not, and below want to try out an alternative.

The obvious source of all the aforementioned difficulties is the fact that a charged particle carries the Coulomb field, which extends to spatial infinity. On the other hand, QFT relies heavily on the idea of locality. Now, these two ingredients are hard to be brought to a peaceful coexistence. Locality implies that  electromagnetic fields of Coulomb-like decay produce flux at infinity which is superselected -- it commutes with all local observables. On the other hand charges with differing velocities produce different fluxes. To stay within one superselection sector one has to ``dress'' particles with free infrared-singular ``clouds'' of radiation which compensate changes of flux. There are arguments within local algebraic approach to QFT \cite{bu82} that by this procedure one can also force the causal support of the charged particle into a~spacelike cone, which means that the particle (together with the surrounding electromagnetic field) is created by an operation supported in this region. It is believed that it should also be possible to base an analysis of quantum statistics on these localization properties.

Another side of the problems we are here concerned with is the question: ``in front of'' what background charged particles appear? That is to say, what is the background representation space of radiation to which operations creating charged particles are applied? The most standard answer is: the vacuum representation. However, it seems that the use of other, infrared singular, representations of radiation may have advantages over the vacuum representation. Such ``infravacuum'' representations have been investigated since 1970'. In fact, one of the main contributors in this field expressed the hope that charged particles could be ` ``ordinary'' particles, but moving in an ``infravacuum'' ' \cite{kr83}. It seems that this is not possible, in the sense of the existence of a discrete mass hyperboloid in the energy-momentum spectrum, as shown by the analysis of the Gauss law \cite{bu86}. Nevertheless, the following question is still valid: how much of the infrared structure may be transferred from the particle to the background electromagnetic field? In this article we want to investigate how much of the local structure of fields may be abandoned, to still have a reasonable notion of an asymptotic field, in absence of discrete masses in the energy-momentum spectrum  and of the vacuum state in the representation space. We hope the scheme will have relevance for electrically charged particles, but we think it is of interest irrespective of the answer to this question.

In Section \ref{asymp} we propose to consider the possibility that there exists a~class of charge-creating fields whose (anti-)commutators decay (rather mild\-ly) in spacelike directions, but no further assumptions on their locality with respect to observables like electromagnetic field are needed. We then use a~modified Haag-Ruelle type of the definition of an asymptotic field. However, our approach is based on direct averaging on hyperboloids $x\cdot x=\la^2$, $x^0>0$, without the use of solutions of the Klein-Gordon equation. This method, which was introduced in \cite{he95} (at the level of first-quantized Dirac field), has the advantage of manifest Lorentz invariance. In Section \ref{foutr} we obtain relation of such fields to the fields of the form
\begin{equation}\label{intas}
 \int\wch{\Psi}(p)\wh{\vp}(p)e^{i\big(\sqrt{p^2}-m\big)\la}\,d^4p\,,
\end{equation}
where $\Psi(x)$ is a spacetime translation of the quantum field operator $\Psi$. Our tentative definition of an infraparticle of mass $m$ involves the condition that the above operator converges weakly to a non-zero operator (between states of finite energy), when $\la\to\infty$, for $\wh{\vp}$ supported in some neighborhood of the mass hyperboloid. Using the language of a recent discussion by Dybalski \cite{dy10} of the spectrum of automorphism groups of an algebra, this should amount to a~singular continuous component of the spectrum concentrated on the mass hyperboloid (but we do not go into further discussion of this). One obtains then creation/annihilation operators of an asymptotic field, which transfer energy-momentum lying strictly on the mass hyperboloid. With some strengthening\pagebreak[3] of assumptions one also obtains bosonic/fermionic statistics. We use some of the techniques of the Haag-Ruelle scattering theory \mbox{\cite{ha58} - \cite{dy05}}, but our context is more general.

We remark that construction similar to \eqref{intas} was used in \cite{fr96} to obtain a~Haag-Ruelle-type theory for plektons (in $2+1$ dimensions). However, in the present paper the fields defined by \eqref{intas} are not a starting point, but rather they appear as a result of asymptotic expansion.

What are physically motivated, general sufficient conditions for a theory to admit non-zero operator limits defined above is an important problem for future research. However, an important part of our motivation is the existence of a model which provides an adequate context for the above ideas.
We expect that the algebraic model proposed in \cite{he98} and further developed in \cite{he08}, \cite{he11} is a good candidate for algebra of asymptotic fields in QED. The matter and electromagnetic fields of the model are as far decoupled, as the long-range structure allows: remnant correlations manifest themselves in non-commutation of these fields, and Gauss' law is respected. In particular, the charged field of the model would be then expected to be the result of a~limiting procedure of the type mentioned above, so repeated limiting should satisfy the above structure. We indicate in Section \ref{model} that this is indeed the case. Let us also mention here that the electromagnetic field of the model should also be the result of some asymptotic limiting. However, it seems that the procedures discussed for such purpose in \cite{bu77} and \cite{bu82'} fall short in this case. The reason is twofold: (i) the space of test fields of the electromagnetic field of the model is substantially larger than usual, including a class of non-Schwartz smearing functions, and (ii) there are indications \cite{he11} that the model does not admit spacelike-cone localization of charged electromagnetic fields, which prevents application of the technique used in \cite{bu82'}. We do not address these problems in this article.

The Haag-Ruelle scattering theory assumes the existence of the vacuum vector state and a discrete mass hyperboloid. Our relaxed assumptions on the spacelike decay of (anti-)commutators and our definition of smearing may be also used in this case. In Section \ref{hrtheory} we show that this leads to an extension of the applicability, with a simultaneous sharpening of the results of this formalism. An extension of the Haag-Ruelle theory to quasi-local fields was given earlier in \cite{fa73}, but on much more restrictive basis: existence of mass gap in the energy-momentum spectrum and fast decrease of nonlocalities. A later discussion by one of the authors \cite{sc10}, with more general assumptions, refers to the scattering theory in the spirit of Buchholz \cite{bu91} (see remarks above), rather than Haag-Ruelle.

\pagebreak

Appendix \ref{regwave} contains some more sharp, than usually discussed, results on regular wave-packets satisfying Klein-Gordon  equation. These properties are needed in Section \ref{foutr}. In Appendix \ref{deccom} we state a decay property of correlations of (anti-)commutators for local fields.

\section{Asymptotic relations}\label{asymp}

With a choice of an origin, the spacetime becomes the Minkowski vector space $(M,g)$, with $x\cdot y$ and $x^2=x\cdot x$ notation for scalar products. The interior of the future light-cone will be denoted by $V_+$, its closure by $\clc$, and $H_+$ will stand for the future branch of the unit hyperboloid $x^2=1$. For $v\in H_+$ the invariant measure $d^3v/v^0$ will be denoted by $d\mu(v)$. The scalar product and norm for $f,g$ in the Hilbert space $L^2(H_+,d\mu)$ will be denoted by $(f,g)_{H_+}$ and $\|f\|_{H_+}$ respectively. If a Minkowski basis $(e_0,\ldots,e_3)$ is chosen, we shall denote by $\vec{x}$ the orthogonal projection of $x\in M$ onto the subspace orthogonal to $e_0$, with $\vec{x}\cdot\vec{y}$ denoting the Euclidean scalar product in this subspace and $|\vec{x}|$ the norm. Then  $|x|^2=|x^0|^2+|\vec{x}|^2$. The Lebesgue measure element in $M$ will be denoted by $dx$.

We assume that a QFT is defined in terms of a field *-algebra $\F$ of bounded operators acting in a Hilbert space $\Hc$. The algebra includes, beside observables, also operators interpolating between inequivalent representations of observables, such as creators/annihilators of electric charge. Spacetime translations are performed by a unitary, continuous representation $U(a)$ of the translation group acting in $\Hc$, and the spectrum of its generators is contained in $\clc$ (relativistic energy positivity). However, we do \emph{not} assume the existence of the vacuum vector state, nor the action of a Lorentz group representation in $\Hc$ . For each bounded operator $A$ acting in $\Hc$ and an integrable function $\vp$ one denotes
\begin{equation}\label{field}
A(x)=U(x)AU(-x)\,,\quad A(\vp)=\int\vp(x)A(x)\,dx\,,
\end{equation}
so that
\begin{equation}\label{ffield}
 A(\chi)(x)=\int\chi(y-x)A(y)\,dy\,,\quad A(\chi)(\vp)=A(\chi*\vp)\,,
\end{equation}
where $\chi*\vp$ is the convolution of functions.

In all what follows we consider two versions of the constructions, fer\-mio\-nic or bosonic indicated by subscripts $\pm$.
\begin{ass}\label{com}
The algebra $\F$ contains a subset $\K_\pm$, closed under conjugation, with the following property.
There is a $\kappa>0$ such that for \mbox{$\Psi_1,\Psi_2\in\K_\pm$} the following bounds hold:
\begin{equation}\label{fcom}
 \|[\Psi_1,\Psi_2(a)]_{\pm}\|\leq\frac{c}{(r+|\vec{a}|-|a^0|)^\kappa}\quad\text{for}\quad a^2\leq0\quad(\text{i.e.}\ |\vec{a}|-|a^0|\geq0)\,,
\end{equation}
with some constants $c$ and $r$.

The assumption is covariant: if the bound holds in any particular reference system, it is valid in all other, with some other constant $c$.
\end{ass}
\noindent
The covariance of the condition follows from the relations
\begin{equation*}
 |\vec{a}|-|a^0|=-a^2/(|\vec{a}|+|a^0|)\,,\quad \al^{-1}\leq(|\vec{a}|+|a^0|)(|\vec{a}'|+|a'^0|)^{-1}\leq \al\,,
\end{equation*}
where primed quantities refer to another Minkowski basis and $\al>1$ is a~constant depending on the relation between the bases.

We note that as $[\Psi_1(x),\Psi_2(y)]_\pm=U(x)[\Psi_1,\Psi_2(y-x)]_\pm U(-x)$, one has $\|[\Psi_1(x),\Psi_2(y)]_\pm\|=\|[\Psi_1,\Psi_2(y-x)]_\pm\|$.
Thus the fields $\Psi(x)$ need not be local, but the interference of operations performed by their application decreases with the spacelike distance.

An important fact about the bounds of the above Assumption is that they are conserved under smearing:

\begin{pr}\label{boundchi}
 If $\chi_i$ are Schwartz functions and $\Psi_i$ satisfy Assumption \ref{com} ($i=1,2$), then also $\Psi_i(\chi_i)$ satisfy the bounds \eqref{fcom} (possibly with some other constant $c$).
\end{pr}
\begin{proof}
By a change of integration variables one shows that for each $n>4$:
\begin{multline}\label{fboundchi}
 \|[\Psi_1(\chi_1),\Psi_2(\chi_2)(a)]_\pm\|\leq\int\rho_{12}(z)\|[\Psi_1,\Psi_2(a+z)]_\pm\|\,dz\\
 \leq\con\int\frac{\|[\Psi_1,\Psi_2(a+z)]_\pm\|}{(r+|z|)^n}\,dz\,,
\end{multline}
where $\rho_{12}(z)=\int|\chi_1(w)\chi_2(w+z)|\,dw$, and in the second step we used the fact that $\rho_{12}$ is of fast decrease. Let now $a^2\leq0$ and split the integration domain in the rhs of \eqref{fboundchi} into two sets: (i) $|z|\leq(|\vec{a}|-|a^0|)/4$, and (ii) the rest. In domain (ii) we use the fact that  $\|[\Psi_1,\Psi_2(b)]_\pm\|\leq2\|\Psi_1\|\|\Psi_2\|$ for all~$b$, so this region contributes a term bounded by $\con(r+|\vec{a}|-|a^0|)^{-n+4}$; now $n\geq4+\kappa$ is enough for the thesis. In the domain (i) one has
\begin{equation*}
 |\vec{a}+\vec{z}|-|a^0+z^0|\geq|\vec{a}|-|\vec{z}|-|a^0|-|z^0|\geq(|\vec{a}|-|a^0|)/2\,,
\end{equation*}
so $\|[\Psi_1,\Psi_2(a+z)]_\pm\|$ is bounded by the rhs of \eqref{fcom} multiplied by $2^\kappa$; this again is sufficient for the thesis.
\end{proof}

We now introduce another type of smearing of $\Psi$, used earlier for classical fields in \cite{he95}.
\begin{dfn}\label{smear}
For $\Psi\in\K_\pm$, $\la>0$, and a Schwartz function $f$ on $H_+$, we denote
\begin{equation}\label{fsmear}
 \Psi[\la,f]=\Big(\frac{\la}{2\pi}\Big)^{3/2}\int \Psi(\la v)f(v)\,d\mu(v)\,.
\end{equation}
\end{dfn}
\noindent
\begin{thm}\label{limit}\ \\
(i) Let $\kappa>3$ in Assumption \ref{com}, what we assume from now on. Then
\begin{equation}\label{flimit}
 \limsup_{\la\to\infty}\big\|\big[\Psi_1[\la,f_1],\Psi_2[\la,f_2]\big]_\pm\big\|\leq\con\int|f_1(v)f_2(v)|(v^0)^3\,d\mu(v)\,.
\end{equation}
(ii) Let in addition the supports of $f_i$ be disjoint, and denote $\cosh\gamma_{12}=\inf\{v_1\cdot v_2\mid v_i\in\supp f_i\}>1$.
Then for any $0\leq\gamma<\gamma_{12}$ the bound
\begin{equation}\label{disj}
 \big\|\big[\Psi_1[\la_1,f_1],\Psi_2[\la_2,f_2]\big]_\pm\big\|\leq\frac{\con}{(\la_1\la_2)^{(\kappa-3)/2}}
\end{equation}
holds uniformly for $\exp(-\gamma)\leq\la_1/\la_2\leq\exp(\gamma)$.
\end{thm}
\begin{proof}
(ii) Denote $\la=\sqrt{\la_1\la_2}$, $s=\sqrt{\la_1/\la_2}$, so $\la_1v-\la_2u=\la(sv-s^{-1}u)$. If $\exp(-\gamma/2)\leq s\leq\exp(\gamma/2)$, then $-(sv-s^{-1}u)^2\geq2(\cosh\gamma_{12}-\cosh\gamma)>0$. As at the same time $|s\vec{v}-s^{-1}\vec{u}|+|sv^0-s^{-1}u^0|\leq2\exp(\gamma/2)(v^0+u^0)$, so $|s\vec{v}-s^{-1}\vec{u}|-|sv^0-s^{-1}u^0|\geq(\cosh\gamma_{12}-\cosh\gamma)\exp(-\gamma/2)(v^0+u^0)^{-1}$.  With the use of Assumption \ref{com} one finds then that the lhs of \eqref{disj} is bounded by
\begin{multline}\label{la1la2}
 \con\int\frac{\la^3\,|f_1(v)||f_2(u)|\,d\mu(u)\,d\mu(v)}{[r+\la(|s\vec{v}-s^{-1}\vec{u}|-|sv^0-s^{-1}u^0|)]^\kappa}\\
 \leq\con\,\la^{-(\kappa-3)}\int|f_1(v)||f_2(u)|(v^0+u^0)^\kappa\,d\mu(u)\,d\mu(v)\,,
\end{multline}
which ends this part of the proof.\pagebreak\\
(i) It is easy to see that for $\la_1=\la_2=\la$ ($s=1$) the first form of the bound \eqref{la1la2} is valid irrespective of the support properties of $f_i$. We change the integration variables to $\vec{w}=\la(\vec{u}-\vec{v})$, $\vec{z}=(\vec{u}+\vec{v})/2$, and then this bound becomes
\begin{equation}\label{wz}
 \big\|\big[\Psi_1[\la,f_1],\Psi_2[\la,f_2]\big]_\pm\big\|
 \leq\con\int\frac{|f_1(v)f_2(u)|\,d^3w\,d^3z}{v^0u^0\bigg[r+|\vec{w}|-\dfrac{2|\vec{z}\cdot\vec{w}|}{u^0+v^0}\bigg]^\kappa}\,,
\end{equation}
where $u$ and $v$ are functions of $\vec{w}$, $\vec{z}$ and $\la$. It is now sufficient to show that the limit for $\la\to\infty$ of the rhs of this bound is equal to the rhs of \eqref{flimit}. Elementary calculation shows that $d(u^0+v^0)/d\la\leq0$, so $u^0+v^0\geq\lim_{\la\to\infty}(u^0+v^0)=2z^0$, where we introduced the vector $z\in H_+$ with the space part $\vec{z}$. Thus we have
\begin{equation*}
 \bigg[r+|\vec{w}|-\dfrac{2|\vec{z}\cdot\vec{w}|}{u^0+v^0}\bigg]^{-\kappa}\leq\lim_{\la\to\infty}\Big[...\Big]^{-\kappa}
 =\bigg[r+|\vec{w}|-\dfrac{|\vec{z}\cdot\vec{w}|}{z^0}\bigg]^{-\kappa}\,.
\end{equation*}
Taking into account that $f_1(u)f_2(v)$ is a Schwartz function on $H_+\times H_+$, we can apply the Lebesgue theorem and draw the $\la\to\infty$ limit under the integral. In this way we find that the limit of the rhs of \eqref{wz} is
\begin{equation*}
 \con\int |f_1(z)f_2(z)|\bigg\{\int\bigg[r+|\vec{w}|-\frac{|\vec{z}\cdot\vec{w}|}{z^0}\bigg]^{-\kappa}\,d^3w\bigg\}\frac{d\mu(z)}{z^0}\,.
\end{equation*}
Performing the elementary integral inside the braces one arrives at the thesis.
\end{proof}

\noindent
{\bf Remark}\ \
The functional dependence on $a$ of the bound in Assumption \ref{com} is designed not to enforce steep descent in the neighborhood of the lightcone. Explicitly depending only on invariants, but more restrictive, alternative for the rhs of the bound \eqref{fcom} would be $[r+\sqrt{-a^2}]^{-\kappa}$ for $a^2<0$. One can show that Propositions \ref{boundchi} and \ref{limit} remain in force, with the modification that the rhs in \eqref{flimit} takes the explicitly invariant form $\con\int|f_1(v)f_2(v)|\,d\mu(v)$.

\section{Fourier transforms and a particle}\label{foutr}

We use the following convention for Fourier transforms:
\begin{equation}
 \wh{\chi}(p)=\frac{1}{(2\pi)^2}\int \chi(x) e^{ip\cdot x}\,dx\,,
\end{equation}
and the inverse transform of $\chi$ is denoted by $\wch{\chi}$. The distributional transform is then defined as usually by $\wh{T}(\chi)=T(\wh{\chi})$, which is equivalent to  $\wch{T}(\wh{\chi})=T(\chi)$. We apply this also to the operators defined in \eqref{field} and use the usual symbolic integral notation:
\begin{equation}
 \Psi(\chi)=\wch{\Psi}(\wh{\chi})=\int\wch{\Psi}(p)\wh{\chi}(p)\,dp\,.
\end{equation}
Note also that
\begin{equation}
 \Psi(\chi)^*=\Psi^*(\ov{\chi})=\wch{\Psi^*}(\wh{\ov{\chi}})\,,\quad\text{and}\quad \supp\wh{\ov{\chi}}=-\supp\wh{\chi}\,.
\end{equation}
It is well-known that the momentum transfer of $\Psi(\chi)$ is restricted by the support of $\wh{\chi}$ (see e.g. \cite{ha92}). More precisely, if $E(\Delta)$ projects onto the subspace of $\Hc$ with spectral values of the energy-momentum operators in a Borel set~$\Delta$, then
\begin{equation}
 \Delta_2\cap(\supp\wh{\chi}+\Delta_1)=\varnothing\quad\text{implies}\quad
 E(\Delta_2)\Psi(\chi)E(\Delta_1)=0\,.
\end{equation}

We now apply the smearing defined in \eqref{fsmear} to $\Psi(\chi)$ and find that
\begin{equation}\label{laF}
 \Psi(\chi)[\la,f]=\Psi(F_\la)\,,\quad\text{with}\quad F_\la(x)=\Big(\frac{\la}{2\pi}\Big)^{3/2}\int\chi(x-\la v)f(v)\,d\mu(v)\,,
\end{equation}
and then
\begin{equation}\label{conj}
 \Psi(\chi)[\la,f]^*=\Psi^*(\ov{\chi})[\la,\ov{f}]\,.
\end{equation}
It is obvious that $F_\la$ is a Schwartz function and
\begin{equation}\label{FF}
 \wh{F_\la}(p)=\Big(\frac{\la}{2\pi}\Big)^{3/2}\wh{\chi}(p)\int f(v)e^{i\la p\cdot v}\,d\mu(v)\,.
\end{equation}
One of the consequences of Theorem \ref{limit} is the following corollary.
\begin{pr}\label{normb}
If Assumption \ref{com} with $\kappa>3$ holds, then\\
(i) in fermionic case:
\begin{equation}\label{flsnorm}
 \limsup_{\la\to\infty}\|\Psi(\chi)[\la,f]\|\leq\con\|(v^0)^{3/2}f(v)\|_{H_+}\,,
\end{equation}
(ii) in bosonic case: if $\Delta\subseteq \clc$ is a bounded set and $0\notin\supp\wh{\chi}$, then
\begin{equation}\label{blsnorm}
 \limsup_{\la\to\infty}\|\Psi(\chi)[\la,f]E(\Delta)\|\leq\con\,\|(v^0)^{3/2}f(v)\|_{H_+}\,.
\end{equation}
\end{pr}
\noindent
In the proof we shall use the following result due to Buchholz (\cite{bu90}, Lemma~2.1 and its obvious consequence).
\begin{lem}
Let $A$ be a bounded operator and $P$ the orthogonal projection operator onto the kernel space of $A^n$. Then:
\begin{equation}
 \|AP\|^2\leq(n-1)\|[A,A^*]\|\,,\quad \|A^*P\|^2\leq n\|[A,A^*]\|\,.
\end{equation}
\end{lem}
\noindent
\begin{proof}[Proof of Prop.\ \ref{normb}]
(i) As $\|A^*A\|\leq\|[A^*,A]_+\|$ for any bounded operator, this is a direct consequence of \eqref{conj} and Theorem \ref{limit}.\\
(ii) Suppose first that $\supp\wh{\chi}\subseteq\{p\mid p^0\leq-\delta\}$ for some $\delta>0$. Then the energy-momentum transfer of $[\Psi(\chi)[\la,f]]^n$ is contained in \mbox{$\{p\mid p^0\leq-n\delta\}$}, so $[\Psi(\chi)[\la,f]]^nE(\Delta)=0$ for sufficiently large $n$. With the use of the Lemma the thesis now follows as in (i). If \mbox{$\supp\wh{\chi}\subseteq\{p\mid p^0\geq\delta\}$} similar reasoning holds with the use of $\Psi(\chi)[\la,f]^*$. In the general case a~closed set not containing zero vector may by covered by sets of the two above considered types, with respect to several Minkowski bases (note that for $v\in H_+$ there is $\al^{-1}v^0\leq v^0{}'\leq\al v^0$ for zero-coordinates in two Minkowski bases, with $\al$ independent of $v$).
\end{proof}

From now on we assume that $\supp\wh{\chi}$ is compact and contained in $V_+$. This allows us to apply the result of Proposition \ref{exp} in Appendix~\ref{regwave} for the expansion of $\wh{F_\la}(p)$ in Eqs.\ \eqref{laF} and \eqref{FF}. We apply the operator distribution $\wch{\Psi}(p)$ to both sides of the identity \eqref{fexp}. This immediately gives
\begin{equation}\label{expan}
 e^{-i3\pi/4}\sum_{j=0}^N\la^{-j}\Psi(\chi_j)[\la,f_j]=\int e^{i\la\sqrt{p^2}}\wh{\chi}(p)f\big(p/\sqrt{p^2}\big) \wch{\Psi}(p)\,dp
 +\Psi(R_\la)\,,
\end{equation}
where $\wh{\chi}_0(p)=(p^2)^{3/4}\wh{\chi}(p)$, $f_0=f$ and the other functions are defined in Appendix \ref{regwave}. All functions $f_j$ and $\chi_j$ are smooth, $f_j$ are of compact support, and $\supp\wh{\chi}_j=\supp\wh{\chi}$. This results in the next theorem, which we precede with the following denotation.

Let $\chi$ be a smooth function such that the support of $\wh{\chi}$ is compact and contained in $V_+$. Then $\schi$ will denote the function defined by
\begin{equation}
 \wh{\schi}(p)=(p^2)^{3/4}\wh{\chi}(p)\,.
\end{equation}
We note that the mapping $\chi\mapsto\schi$ is a linear isomorphism of this class of functions.
\begin{thm}\label{asymptotic}
If Assumption \ref{com} with $\kappa>3$ holds, and $\supp\wh{\chi}$ is compact, contained in $V_+$, then the following asymptotic relations hold:
\begin{equation}\label{asrel}
 e^{-i3\pi/4}\Psi(\schi)[\la,f]E(\Delta)
 =\int e^{i\la\sqrt{p^2}}\wh{\chi}(p)f\big(p/\sqrt{p^2}\big) \wch{\Psi}(p)\,dp\,E(\Delta)+O_{\|.\|}(\la^{-1})\,.
\end{equation}
Subscript $\|.\|$ indicates that the respective bounds hold in operator norm. Here and in what follows $\Delta$ is a bounded Borel set in bosonic case, and $E(\Delta)=\1$ in fermionic case.
\end{thm}
\begin{proof}
The estimate \eqref{R} shows that in \eqref{expan}: $\|\Psi(R_\la)\|\leq\con\|\Psi\|\,\la^{-N+3/2}$. The choice $N=3$, with the use of Proposition \ref{normb} for $\Psi(\chi_j)[\la,f_j]$ ($j=1,2,3$), yields the result.
\end{proof}

Knowing that both sides of the relation \eqref{asrel} remain bounded in norm, we can try the following further specification.
\begin{ass}\label{asspart}
We assume that for some $m>0$, all Schwartz functions $f$ on $H_+$, and all smooth $\wh{\chi}$ with compact support in some neighborhood of $mH_+$, there exist weak limits
\begin{multline}\label{fasspart}
 \Psi_m^\out(\schi)[f]E(\Delta)=\mathrm{w}\!-\!\lim_{\la\to\infty}e^{-i(\la m+3\pi/4)}\Psi(\schi)[\la,f]E(\Delta)\\
 = \mathrm{w}\!-\!\lim_{\la\to\infty}\int e^{i\la\big(\sqrt{p^2}-m\big)}\wh{\chi}(p)f\big(p/\sqrt{p^2}\big) \wch{\Psi}(p)\,dp\,E(\Delta)\,.
\end{multline}
\end{ass}

Let $h$ be a smooth real function on $(0,\infty)$, with compact support, and such that $\int h(\la)\,d\la=1$. Denote $\ti{h}(\w)=\int e^{i\w\la}h(\la)\,d\la$. Now, for $\La>0$, multiply the expressions under the limit signs in Eq.\ \eqref{fasspart} by
\begin{equation}\label{hLa}
 h_\La(\la)=\La^{-1}h(\la/\La)\,,
\end{equation}
and integrate over $\la$. This gives
\begin{multline}\label{asspartsm}
 \Psi_m^\out(\schi)[f]E(\Delta)=\mathrm{w}\!-\!\lim_{\La\to\infty}\int h_\La(\la)e^{-i(\la m+3\pi/4)}\Psi(\schi)[\la,f]\,d\la\, E(\Delta)\\
 =\mathrm{w}\!-\!\lim_{\La\to\infty}\int\ti{h}\big(\La[\sqrt{p^2}-m]\big)\wh{\chi}(p)f\big(p/\sqrt{p^2}\big) \wch{\Psi}(p)\,dp\,E(\Delta)\,.
\end{multline}
\begin{lem}
If $\supp\wh{\chi}\cap mH_+=\varnothing$, then $\Psi_m^\out(\schi)[f]=0$.
\end{lem}
\begin{proof}
For any $\vp_1,\vp_2\in\Hc$ the product $(\vp_1,\wch{\Psi}(p)\vp_2)$ defines a tempered distribution, so may be represented by continuous functions and their distributional derivatives. Therefore all contributions to the operator under the limit on the rhs of \eqref{asspartsm}, placed between $\vp_1,\vp_2$, are of the form
\begin{equation*}
 \La^k\int [D^\alpha\ti{h}]\big(\La[\sqrt{p^2}-m]\big)D^\beta\wh{\chi}(p)F(p)\,dp
\end{equation*}
in standard multi-index notation, where $F$ is a continuous function. This is easily shown to vanish faster than any inverse power of $\La$ for $\La\to\infty$, if the premise holds.
\end{proof}

\begin{dfn}\label{out}
Suppose Assumption \ref{asspart} is satisfied with nonzero limit operators. Let $f$ be of compact support $\supp f\subset H_+$, and let compactly supported~$\wh{\chi}$ satisfy $\wh{\chi}(p)=1$ in some neighborhood of the set $m\supp f\subset mH_+$.\\
Then we set
\begin{multline}
 \Psi_m^\out[f]E(\Delta)=\mathrm{w}\!-\!\lim_{\la\to\infty}e^{-i(\la m+3\pi/4)}\Psi(\schi)[\la,f]E(\Delta)\\
 =\mathrm{w}\!-\!\lim_{\la\to\infty}\int e^{i\la\big(\sqrt{p^2}-m\big)}\wh{\chi}(p)f\big(p/\sqrt{p^2}\big) \wch{\Psi}(p)\,dp\,E(\Delta)
\end{multline}
and we interpret this as a creation operator of an asymptotic particle with mass~$m$.
\end{dfn}\noindent
Note that due to the preceding lemma the definition is indeed independent of $\chi$ in the assumed class.

\begin{pr}
Let the assumptions leading to Definition \ref{out} be satisfied. Then the following holds:
\begin{gather}
 \|\Psi_m^\out[f]E(\Delta)\|\leq\con\|(v^0)^{3/2}f(v)\|_{H_+}\,,\label{psibound}\\[1.5ex]
 \Delta_2\cap(m\supp f+\Delta_1)=\varnothing\quad\text{implies}\quad
 E(\Delta_2)\Psi_m^\out[f]E(\Delta_1)=0\,.\label{transf}
\end{gather}
Moreover, if $\supp f_i\subseteq D_\nu=\{v\in H_+\mid |\vec{v}|\leq\nu\}$ ($i=1,2$), then
\begin{equation}\label{difference}
 \supp f_1-\supp f_2\subseteq D_\nu-D_\nu\subset\Big\{\,q\,\,\big|\,\,|q^0|/|\vec{q}|\leq\frac{\nu}{\sqrt{\nu^2+1}}\,,\ |\vec{q}|\leq2\nu\Big\}\,.
\end{equation}
In this case \mbox{$\Delta_2\cap\big(m(D_\nu-D_\nu)+\Delta_1\big)=\varnothing$} implies
\begin{equation}\label{pptransf}
 E(\Delta_2)\Psi_{1m}^\out[f_1]^*\Psi_{2m}^\out[f_2]E(\Delta_1)=
 E(\Delta_2)\Psi_{1m}^\out[f_1]\Psi_{2m}^\out[f_2]^*E(\Delta_1)=0
\end{equation}
\end{pr}
\begin{proof}
The estimate \eqref{psibound} is a direct consequence of Proposition \ref{normb}. To prove the statement \eqref{transf} let first $\Delta_1,\Delta_2$ be compact and satisfy the assumption. Then one can find an open neighborhood $U$ of $m\supp f$ such that \mbox{$\Delta_2\cap(U+\Delta_1)=\varnothing$}. The implication follows in this case by choosing  $\wh{\chi}$ in the class defining $\Psi_m^\out[f]$, with support in $U$. The general case now follows by regularity of spectral measures. Relations \eqref{pptransf} may be shown in similar way. The estimate in \eqref{difference} follows easily from the obvious identity for $v,u\in H_+$: $v^0-u^0=(\vec{v}-\vec{u})\cdot(\vec{v}+\vec{u})/(v^0+u^0)$ .
\end{proof}

\noindent
{\bf Remark} Property \eqref{transf} is decisive for the interpretation: creation of an asymp\-to\-tic particle adds energy-momentum strictly on the mass hyperboloid. However, this need not be reflected in the presence of a discrete mass hyperboloid in the energy-momentum spectrum, if there is no vacuum in the Hilbert space.
\pagebreak[2]

More can be inferred if one adds a stronger assumption on the nature of convergence in Definition \ref{out} with added smearing of the form leading to  \eqref{asspartsm}.
\begin{ass}\label{outstr}
We assume that under the conditions of Definition \ref{out} the field $\Psi_m^\out[f]$ defined there may be obtained as a strong limit:
\begin{multline}
 \Psi_m^\out[f]E(\Delta)=\mathrm{s}\!-\!\lim_{\La\to\infty}\int h_\La(\la)e^{-i(\la m+3\pi/4)}\Psi(\schi)[\la,f]\,d\la\, E(\Delta)\\
 =\mathrm{s}\!-\!\lim_{\La\to\infty}\int\ti{h}\big(\La[\sqrt{p^2}-m]\big)\wh{\chi}(p)f\big(p/\sqrt{p^2}\big) \wch{\Psi}(p)\,dp\,E(\Delta)\,,
\end{multline}
where $h_\La$ is defined by \eqref{hLa}, with $h$ any function in the class defined there.
\end{ass}

\begin{pr}\label{commut}
 If Assumption \ref{outstr} is satisfied, then $\supp f_1\cap\supp f_2=\varnothing$ implies
\begin{gather}
 \big[\Psi_{1m}^\out[f_1]^\sharp,\Psi_{2m}^\out[f_2]^\sharp\big]_\pm E(\Delta)=0\,,\label{comout}\\[1ex]
 \Big[\Psi_{1m}^\out[f_1]^\sharp,\big[\Psi_{2m}^\out[f_2]^\sharp,\Psi_{3m}^\out[f_3]^\sharp\big]_\pm\Big] E(\Delta)=0\label{com3out}\,.
\end{gather}
where $\sharp$ is either empty or the adjoint operator star $*$ (uncorrelated at the two or three operators).
\end{pr}
\begin{proof}
First note that for the given bounded $\Delta$ there is some other bounded $\Delta'$ such that
\begin{multline*}
 \Psi_i(\schi_i)[\la_i,f_i]^\sharp\Psi_j(\schi_j)[\la_j,f_j]^\sharp E(\Delta)\\
 =\Psi_i(\schi_i)[\la_i,f_i]^\sharp E(\Delta')\Psi_j(\schi_j)[\la_j,f_j]^\sharp E(\Delta)\,.
\end{multline*}
Therefore
\begin{multline*}
 \big\|\big[\Psi_{1m}^\out[f_1]^\sharp,\Psi_{2m}^\out[f_2]^\sharp\big]_\pm E(\Delta)\big\|\\
 \leq\lim_{\La\to\infty}\int|h_1(\xi_1)h_2(\xi_2)|\big\|\big[\Psi_1(\schi_1)[\La\xi_1,f_1]^\sharp,\Psi_2(\schi_2)[\La\xi_2,f_2]^\sharp\big]_\pm\big\|\,d\xi_1\,d\xi_2\,.
\end{multline*}
Let now supports of $f_i$ be separated as in Theorem \ref{limit} (ii). But the supports of $h_i$ may be chosen such that $\exp(-\gamma)\leq\xi_1/\xi_2\leq\exp(\gamma)$ in the notation of this theorem. Then the expression under the limit vanishes as $\La^{-(\kappa-3)}$, which gives the first relation of Proposition. For the second relation it is sufficient to note that one can decompose $f_3=f_{31}+f_{32}$, where the support of $f_{3i}$ does not intersect that of $f_i$. Then the thesis follows for $f_{32}$ directly, and for $f_{31}$ by Poisson's identity (or similar identity in fermionic case).
\end{proof}

\noindent
{\bf Remark} Property \eqref{comout} is responsible for appropriate fermionic/bosonic statistics of the asymptotic particles. Property \eqref{com3out} generalizes  a part of the structure needed for asymptotic Fock space construction (see below the Haag-Ruelle case). However, in absence of vacuum situation is more complex. We do not study this question in generality here.

\section{A model}\label{model}

As mentioned in Introduction, our principal motivation is the study of the charged particle problem in quantum electrodynamics. Whether the concepts introduced in preceding sections will be of relevance for that case is not known yet. However, here we want to point out that an algebraic model put forward some time ago \cite{he98} as a candidate for the algebra of asymptotic fields in QED fits into the scheme. The model includes no dynamical interdependence of the matter and radiation fields, but there do exist remnant correlations between them which ensure the validity of a form of Gauss' law~\cite{he98}. Also, we want to emphasize the importance of the choice of hyperboloids rather than constant time hypersurfaces for the limiting procedure. For classical fields, this type of limiting was shown in \cite{he95} to be applicable also in the long-range case  (by an appropriate choice of electromagnetic gauge). This fact was one of the constituents in the construction of the model, and we expect this property to survive quantization.

The algebra of the model given in \cite{he98} does not refer to any spacetime localization of fields. In later articles \cite{he08} and \cite{he11} two alternative versions of localization were formulated: in regions restricted spatially, but extending to timelike infinity, or in generalized time-slices (linearly fattening towards spatial infinity). In either of this cases the algebra of the model may be given as follows:
\begin{equation}\label{stalg}
 \begin{split}
    W_\as(J)^*&=W_\as(-J)\,,\quad W_\as(0)=E\,,\\[1ex]
    W_\as(J_1)W_\as(J_2)&=\exp[-\tfrac{i}{2}\{J_1,J_2\}]W_\as(J_1+J_2)\,,\\[1ex]
    [\psi_\as(\chi_1),\psi_\as(\chi_2)]_+&=0\,,\qquad
    [\psi_\as(\chi_1),\psi_\as(\chi_2)^*]_+=\<\chi_1,\chi_2\>\,,\\[1ex]
    W_\as(J)\psi_\as(\chi)&=\psi_\as(\chi')W_\as(J)\,,\quad\text{where}
    \quad [\chi']=S_J[\chi]\,.
 \end{split}
\end{equation}
The elements $W_\as(J)$ describe the exponentiated electromagnetic field, and elements $\psi_\as(\chi)$ form a free, in the sense of the field equation, Dirac field. The smearing functions $J$ and $\chi$ are not, in general, of compact support. The scalar product for Dirac fields $\<\chi_1,\chi_2\>$ is standard, but the symplectic form for electromagnetic fields $\{J_1,J_2\}$ is a nontrivial extension of the standard form to a larger function space. The most important constitutive element of the structure is the presence of the nontrivial linear automorphisms $S_J$ of the space of equivalence classes of matter test functions (functions in one class, denoted by $[\chi]$,  give rise to the same element of the algebra), which define the non-commutativity of Dirac and electromagnetic fields. This non-commutativity prevents the model from complete factorization, in any Hilbert space representation, into matter $\times$ radiation structure. Also, no physical vacuum representation is admitted by the model~\cite{he98}.

A large class of physically motivated representations of the above algebra has the following properties. The Hilbert space of the representation has the form $\Hc=\Hc_F\otimes\Hc_r$, where $\Hc_F$ is the standard Fock space of the Dirac field, which is represented in $\Hc_F$ in standard vacuum representation.  Translations are represented by unitary group of operators $U(x)=U_F(x)\otimes U_r(x)$. But $\Hc_r$ is \emph{not a vacuum Fock space} and the electromagnetic field is \emph{not represented in $\Hc_r$ alone}. Thus the structure \emph{does not factorize} into tensor product. The spectrum of $U(x)$ is contained in $\clc$, but includes no discrete hyperboloid, and there is no vacuum vector in the total representation space. (We refer the reader to the original articles for more details and interpretation.)

The assumptions of our present constructions are satisfied immediately in these representations. Assumption \ref{com} is easily verified for spaces of test functions $\chi$ used in \cite{he08} and \cite{he11} (in the first case the functions can be made to vanish arbitrarily fast at infinity, and in the second case they are Schwartz functions). Assumptions \ref{asspart} and \ref{outstr} are satisfied trivially: if $\psi_\as(\chi)$ is substituted for $\Psi$ in \eqref{fasspart}, the operator under the limit on the far rhs does not depend on~$\la$. This is because $\psi_\as(\chi)$ commutes with $U_r(x)$, so its translations are as in the free field case.

This is of course a rather trivial application, but it shows that the structure considered here is free from contradiction.  Also, if the model could indeed be derived by some limiting procedure of the type considered in this paper, then this is what one should expect: repeat limiting should be trivial.

\section{Haag-Ruelle case}\label{hrtheory}

The logic of the Haag-Ruelle construction (HR) is somewhat different from the one we follow in this article. The HR formalism is designed to construct an asymptotic Hilbert space which may be generated from the vacuum by the asymptotic fields, with no regard to the question of asymptotic completeness; this space may be a proper subspace of the Hilbert space of the theory. Therefore, one does not need a condition engaging, like our Assumption \ref{asspart}, the whole Hilbert space of the theory for the limiting hypothesis. Instead, the existence of vacuum and of a discrete hyperboloid in the energy-momentum spectrum supply a more specific setting, in which the existence of asymptotic fields may be proved on an asymptotic (sub-)space.

In this section we want to indicate that a manifestly Lorentz-invariant variation of the HR construction may be based on the averaging on hyperboloids introduced by Definition~\ref{smear}. For the derivation of scattering states and statistics our Assumption \ref{com} (instead of strict locality) is sufficient. For the derivation of the Fock structure of asymptotic states we engage an additional Assumption \ref{cluster} on a cluster property of (anti-)commutators (see Paragraph \ref{hrasum} below). This assumption is satisfied automatically for local and almost local fields (translations of local fields smeared with Schwartz functions), but in general its derivation from Assumption~\ref{com} has not been proven. On the other hand, the decay of the correlations contained in Assumption \ref{cluster} may be very weak, which may suggest that such derivation from Assumption~\ref{com}, or similar condition, should be possible. This is an open problem.

In addition to widening the applicability of the HR theory, our construction brings also a few refinements, which apply, in particular, to the orthodox -- strictly local -- HR case: weaker spectral condition (formula \eqref{singul}) and lack of a non-covariant and rather unphysical condition on $3$-momentum-space behavior of test functions (appropriate vanishing for $\vec{p}\to0$, see \cite{dy05}). The HR case, of course, does not refer to the infraparticle problem; we include it in our discussion as another testing ground for our method.

\paragraph{}
\label{hrpsi}
We note first that everything up to, and including Theorem \ref{asymptotic}, remains valid. Therefore, if we denote
\begin{align}
 &\Psi_\La[f]=\int h(\xi)e^{-i(\La\xi m+3\pi/4)}\Psi(\schi)[\La\xi,f]\,d\xi\,,\\
 &\Psi'_\La[f]=\int\ti{h}\big(\La[\sqrt{p^2}-m]\big)\wh{\chi}(p)f\big(p/\sqrt{p^2}\big) \wch{\Psi}(p)\,dp\,,
\end{align}
then $(\Psi_\La[f]-\Psi'_\La[f])E(\Delta)=O_{\|.\|}(\La^{-1})\,.$ For brevity, the notation omits the dependence of these operators on $\chi$ and $h$. We assume that $\supp\wh{\chi}$ is contained in some neighborhood of $mH_+$ and $\wh{\chi}(p)=1$ in some smaller neighborhood of \mbox{$m\supp f\subseteq mH_+$}.

\paragraph{}
\label{hrcom}
Let $\supp f_1\cap\supp f_2=\varnothing$. Then for $h_i$ with appropriate supports there~is
\begin{gather}
 \lim_{\La\to\infty}\big\|\big[\Psi_{1\La}[f_1]^\sharp,\Psi_{2\La}[f_2]^\sharp\big]_\pm E(\Delta)\big\|=0\,,\label{limcom2}\\[1ex]
 \lim_{\La\to\infty}\Big\|\Big[\Psi_{1\La}[f_1]^\sharp,\big[\Psi_{2\La}[f_2]^\sharp,\Psi_{3\La}[f_3]^\sharp\big]_\pm\Big] E(\Delta)\Big\|=0\,.\label{limcom3}
\end{gather}
The limits do not change if some of the operators $\Psi_{i\La}[f_i]$ are replaced by their primed versions $\Psi'_{i\La}[f_i]$, and/or their derivatives with respect to $\La$.

Proof closely parallels that of Proposition \ref{commut}, with supports of $h_i$ specified there, the only difference being that here we do not assume the existence of the limits of operators. Admissibility of adding a prime is obvious, while differentiating on $\La$ amounts to the replacement of $h_\La$ by $-\La^{-1}g_\La$, where $g(\la)=d[\la h(\la)]/d\la$ and $g_\La(\la)=\La^{-1}g(\la/\La)$.

\paragraph{}
\label{hrspec}
Now one assumes the existence of the vacuum vector $\W$ in the representation space.
For $0\leq\mu<m$ denote
\begin{equation}
 E_\mu=E\big(\{p\mid |\sqrt{p^2}-m|\leq\mu\,,\ p^0>0\}\big)\,.
\end{equation}
Then for any $0<\ep<m$
\begin{equation}\label{singul}
 \int_0^\ep\|(E_\mu-E_0)\Psi\W\|\frac{d\mu}{\mu}<\infty\quad \text{implies}\quad \int_0^\infty\Big\|\frac{d\Psi'_\La[f]\W}{d\La}\Big\|\,d\La<\infty\,.
\end{equation}

To prove this we denote $R=\sqrt{P^2}-m\1$ and observe that $d\Psi'_\La[f]\W/d\La =R\,\ti{h}'(\La R)F(P)\Psi\W\,,$ where $F$ is of compact support intersecting $mH_+$, and $\ti{h}'(\w)=d\ti{h}(\w)/d\w$. Therefore for some $\delta>0$  there is $\|d\Psi'_\La[f]\W/d\La\|\leq\con\|E_\delta |R|(1+\La^2 R^2)^{-2}\Psi\W\|$. The integral over $\La\in\<0,\delta^{-2}\>$ is obviously finite. For $\La>\delta^{-2}$ we split the rhs into two terms and estimate:
\begin{gather*}
 \|(E_\delta-E_{\La^{-1/2}})|R|(1+\La^2 R^2)^{-2}\Psi\W\|\leq \delta\|E_\delta\Psi\W\|(1+\La)^{-2}\,,\label{Ed}\\[1ex]
 \|E_{\La^{-1/2}}|R|(1+\La^2 R^2)^{-2}\Psi\W\|\leq \con \|(E_{\La^{-1/2}}-E_0)\Psi\W\|\,\La^{-1}\,,\label{EL}
\end{gather*}
where the constant on the rhs of the latter estimate is equal to the maximum of the function \mbox{$|s|(1+s^2)^{-2}$}, $s\in\mR$. Integrating this latter estimate one obtains the implication \eqref{singul}.

\paragraph{} 
\label{hrasst}
Let the above spectral condition be satisfied for all $\Psi\in\K_\pm$. Then for $f_i$ with disjoint supports, and supports of $h_i$ adjusted as in Paragraph \ref{hrcom}, there exist strong limits
\begin{equation}\label{many}
 \mathrm{s}\!-\!\lim_{\La\to\infty}\Psi_{1\La}[f_1]\ldots\Psi_{n\La}[f_n]\W=\mathrm{s}\!-\!\lim_{\La\to\infty}\Psi'_{1\La}[f_1]\ldots\Psi'_{n\La}[f_n]\W\,.
\end{equation}
These limits depend only on the one-operator asymptotic vectors
\begin{equation}\label{one}
 \mathrm{s}\!-\!\lim_{\La\to\infty}\Psi'_{\La}[f]\W=(2\pi)^2f(P/m)E_0\Psi\W\,.
\end{equation}
Thus the structure is nontrivial if, and only if, $E_0\neq0$, i.e.\ there is a discrete mass hyperboloid in the spectrum of energy-momentum, and there exist $\Psi$ which interpolate between $\W$ and $E_0\Hc$.

The existence of the limit is shown by standard Cook method, with the use of Paragraphs \ref{hrcom} and \ref{hrspec}. The one-operator limit is rather obvious, and then the second part is shown by (anti-)commuting (with the use of \eqref{limcom2}) a~particular operator to the right, to stand at the vector $\W$.

\paragraph{} 
\label{hreta}
To obtain the Fock structure of asymptotic states one needs two additional elements. The first is the following generalization. For $\eta\in(0,1\>$ let $s(\La)=(m\La)^\eta/m$ (so that $s$ is a length and $s=\La$ for $\eta=1$).  Denote
\begin{equation}
 h^\eta_\La(\la)=s^{-1}h(s^{-1}(\la-\La)+1)
 \end{equation}
and define $\Psi_\La^\eta[f]$ similarly as $\Psi_\La[f]$, with $h_\La$ replaced by $h^\eta_\La$ (so, in particular $h_\La=h^1_\La$ and $\Psi_\La^1[f]=\Psi_\La[f]$). Then for $\eta<1$ there is
\begin{equation}\label{LL'}
 \mathrm{s}\!-\!\lim_{\La\to\infty}\Psi^\eta_{1\La}[f_1]\ldots\Psi^\eta_{n\La}[f_n]\W= \mathrm{s}\!-\!\lim_{\La\to\infty}\Psi_{1\La}[f_1]\ldots\Psi_{n\La}[f_n]\W\,,
\end{equation}
with arbitrary choice of functions $h_i$ on the lhs, while on the rhs the supports of functions $h_i$ respect the demands of Paragraph \ref{hrasst}.

To justify this we note that by the results of Paragraph \ref{hrasst} we are free to choose the support of functions $h_i$ on the rhs in arbitrarily small interval \mbox{$\<1-\delta,1+\delta\>$}. Then by inspection of the proof of Proposition \ref{commut} one can show that limits \eqref{limcom2} and \eqref{limcom3} are valid for any choice of single $\Psi$-operators appearing on both sides of the above relation (or their adjoints), if $\delta$ is sufficiently small. Then the proof of \eqref{LL'} follows the idea by Buchholz (\cite{dy05}, Lemma 2.4) (with a minor necessary refinement: as all fields $\Psi^\eta_\La[f]$ act on vector states of bounded energy, their norms may be assumed bounded).

\paragraph{} 
\label{hrasum}
The second element used for the derivation of the Fock structure is the following assumption. We denote by $E_\W^\bot$ the projection onto the orthogonal complement of the linear span of $\W$. We also make the following observation and introduce the function $K$:
\begin{multline}\label{4psi}
 (\W,[\Psi_1(x_1),\Psi_2(x_2)]_\pm E_\W^\bot[\Psi_3(x_3),\Psi_4(x_4)]_\pm\W)\\
 =(\W,B_{12}(x_1-x_2)E_\W^\bot  U\big(-\tfrac{1}{2}(x_1+x_2-x_3-x_4)\big)B_{34}(x_3-x_4)\W)\\
 \equiv K(x_1-x_2,x_3-x_4,\tfrac{1}{2}(x_1+x_2-x_3-x_4))\,,
\end{multline}
where $B_{ij}(z)=[\Psi_i(z/2),\Psi_j(-z/2)]_\pm$.
\begin{ass}\label{cluster}
Let $\Psi_i\in\K_\pm$, $i=1,\ldots,4$, and $N$ be any positive integer. Then for large enough, positive $d$, and
\begin{equation}
 |y_1|\leq d\,,\quad |y_2|\leq d\,,\quad |\vec{y}|\geq|y^0|+c_1d\,,\\
\end{equation}
the following estimate holds
\begin{equation}
 |K(y_1,y_2,y)|\leq c_2\frac{d^M}{(|\vec{y}|-|y^0|)^\vep}+c_3d^{-N}\,,
\end{equation}
and the positive constants $c_i$, $M$ and $\vep$ do not depend on $d$.

The assumption is covariant: if it holds in any particular reference system, it is valid in all other, with some other constants $c_i$.
\end{ass}

\begin{pr}\label{clusterchi}
Assumption \ref{cluster} is closed with respect to smearing of fields $\Psi$ with Schwartz functions; more precisely, it remains valid, with some other constants $c_i$,  under replacement $\Psi_i\rightarrow\Psi_i(\chi_i)$.
\end{pr}
\begin{proof}
 Replacing $\Psi_i$ by $\Psi_i(\chi_i)$ in \eqref{4psi} amounts to the replacement $K\rightarrow K*\vp$,
\begin{equation*}
 (K*\vp)(y_1,y_2,y)=\int K(y_1-z_1,y_2-z_2,y-z)\varphi(z_1,z_2,z)\,dz_1dz_2dz\,,
\end{equation*}
where $\varphi$ is a Schwartz function. Let
\begin{equation*}
 |y_1|\leq d'\,,\ |y_2|\leq d'\,,\ |\vec{y}|-|y^0|\geq c'_1d'\,,\
 |z_1|\leq d'\,,\ |z_2|\leq d'\,,\ |z|\leq d'
\end{equation*}
-- integration over the rest of the domain of $z_i,z$-variables gives a $d'{}^{-N}$ contribution. Set $d'=d/2$ and choose $c'_1\geq\max\{4, 4c_1\}$. Then
\begin{gather*}
 |y_1-z_1|\leq d\,,\quad |y_2-z_2|\leq d\,,\\
 |\vec{y}-\vec{z}|-|y^0-z^0|\geq \tfrac{1}{2}(|\vec{y}|-|y^0|)+\tfrac{1}{2}(|\vec{y}|-|y^0|-4d')\geq \tfrac{1}{2}(|\vec{y}|-|y^0|)\geq c_1d\,,
\end{gather*}
so by Assumption \ref{cluster}, in this region,
\begin{multline*}
 |K(y_1-z_1,y_2-z_2,y-z)|\leq \frac{c_2\,d^M}{(|\vec{y}-\vec{z}|-|y^0-z^0|)^\vep}+c_3d^{-N}\\
 \leq \frac{2^\ep c_2\,d^M}{(|\vec{y}|-|y^0|)^\vep}+c_3d^{-N}
\end{multline*}
which is sufficient to conclude the proof.
\end{proof}

Assumption \ref{cluster} holds, in particular, for local and almost local fields, as shown by Propositions \ref{decay} (in Appendix \ref{deccom}) and \ref{clusterchi}. Its derivation in general case from Assumption \ref{com} or similar condition is an open problem. We note, however, that the decay it assumes may be very slow (any $\vep>0$).

\paragraph{} 
\label{hrgeom}
We shall need below the following simple geometrical facts. For \mbox{$v_i\in H_+$}, $|\vec{v}_i|\leq\nu$, $\beta=\nu/\sqrt{\nu^2+1}$, and $\la_i>0$, one has
\begin{gather}
 |\la_1v_1^0-\la_2v_2^0|\leq|\la_1-\la_2|+\beta|\la_1\vec{v}_1-\la_2\vec{v}_2|\,,\label{vec0}\\[1ex]
 |\la_1\vec{v}_1-\la_2\vec{v}_2|-|\la_1v_1^0-\la_2v_2^0|\geq (1-\beta)|\la_1\vec{v}_1-\la_2\vec{v}_2|-|\la_1-\la_2|\,.
\end{gather}
If in addition $|\la_1-\la_2|\leq\sigma$, $|\la_1\vec{v}_1-\la_2\vec{v}_2|\geq 2\sigma/(1-\beta)$, then
\begin{equation}
 |\la_1\vec{v}_1-\la_2\vec{v}_2|-|\la_1v_1^0-\la_2v_2^0|\geq \tfrac{1}{2}(1-\beta)|\la_1\vec{v}_1-\la_2\vec{v}_2|\geq\sigma\,.
\end{equation}

Inequality \eqref{vec0} follows from $|\vec{v}_i|/v_i^0\leq\beta$ and
\begin{multline*}
 |\la_1v_1^0-\la_2v_2^0|=\frac{|\la_1^2-\la_2^2+(\la_1\vec{v}_1+\la_2\vec{v}_2)\cdot(\la_1\vec{v}_1-\la_2\vec{v}_2)|}{\la_1v_1^0+\la_2v_2^0}\\
 \leq\frac{|\la_1^2-\la_2^2|}{\la_1+\la_2}+\frac{\la_1|\vec{v}_1|+\la_2|\vec{v}_2|}{\la_1v_1^0+\la_2v_2^0}\,|\la_1\vec{v}_1-\la_2\vec{v}_2|\,,
\end{multline*}
and the other are its consequences.

\paragraph{} 
\label{hr2psi}
With the running assumptions, for sufficiently small $\eta$, there is
\begin{equation}\label{two}
 \lim_{\La\to\infty}\Psi_{1\La}^\eta[f_1]^*\Psi_{2\La}^\eta[f_2]\W=(2\pi)^4(\Psi_1\W,(\ov{f_1}f_2)(P/m)E_0\Psi_2\W)\,\W\,.
\end{equation}
The projection of this equality onto $\W$ follows from relations \eqref{one} and \eqref{LL'}. Thus the relation will be true, if
\begin{equation}
 \lim_{\La\to\infty}\big\|E_\W^\bot\big[\Psi_{1\La}^\eta[f_1]^*,\Psi_{2\La}^\eta[f_2]\big]_\pm\W\big\|=0\,.
\end{equation}
The latter is a consequence of the next lemma. We denote
\begin{equation}
\begin{split}
 \supp h&\subseteq\<\tau_1,\tau_2\>\subset(0,\infty)\,,\\
 \supp h^\eta_\La&\subseteq\<\La_1,\La_2\>=\<\La+(\tau_1-1)s, \La+(\tau_2-1)s\>\,,\\
 \La_2-\La_1&=(\tau_2-\tau_1)s=\tau s\,.
 \end{split}
\end{equation}
\begin{lem}
\begin{equation}
 \lim_{\stackrel{\La\to\infty}{\la_1,\la_2\in\<\La_1,\La_2\>}}\big\|E_\W^\bot\big[\Psi_1(\schi_1)[\la_1,f_1]^*,\Psi_2(\schi_2)[\la_2,f_2]\big]_\pm\W\big\|=0\,.
\end{equation}
\end{lem}
\begin{proof}
To simplify notation we write $\Psi_i(\schi_i)=\Psi^\diamond_i$ and recall that both Assumptions \ref{com} and \ref{cluster} remain valid for $\Psi_i^\diamond$. The vector inside the norm signs involves integration over \mbox{$v_1,v_2\in H_+$}, $|\vec{v}_i|\leq\nu$, for some $\nu>0$. We divide this domain into two parts: (i) $|\la_1\vec{v}_1-\la_2\vec{v}_2|\leq 2\tau s/(1-\beta)$,  and (ii) the rest. In region (ii) we change variables $\vec{v}_1$ to \mbox{$\vec{w}=\la_1\vec{v}_1-\la_2\vec{v}_2$}. Then the norm of this part is bounded by
\begin{equation*}
 \con\,(\la_2/\la_1)^{3/2}\int\limits_{|\vec{w}|\geq2\tau s/(1-\beta)}\frac{d^3w}{(r+|\vec{w}|)^\kappa}\leq\con\,s^{-(\kappa-3)}\,,
\end{equation*}
where we used relations of Paragraph \ref{hrgeom} (with $\sigma=\tau s$) and the decay of the commutator,  Assumption~\ref{com}. The norm squared of part (i) is bounded by
\begin{multline*}
 \con\,(\la_1\la_2)^3\int d\mu(v_1)\ldots d\mu(v_4)\times\\
 \times|(\W,[\Psi^\diamond_1(\la_1v_1),\Psi^\diamond_2(\la_2v_2)^*]_\pm E_\W^\bot[\Psi^\diamond_1(\la_1v_3)^*,\Psi^\diamond_2(\la_2v_4)]_\pm\W)|\,,
\end{multline*}
where integration is restricted to $|\vec{v}_i|\leq\nu$, $|\la_1\vec{v}_1-\la_2\vec{v}_2|\leq 2\tau s/(1-\beta)$, $|\la_1\vec{v}_3-\la_2\vec{v}_4|\leq 2\tau s/(1-\beta)$. Again, we divide this domain into two regions: (iii) $\tfrac{1}{2}|\la_1\vec{v}_1+\la_2\vec{v}_2-\la_1\vec{v}_3-\la_2\vec{v}_4|\leq 4\alpha\tau s$, with the constant $\alpha$ to be specified below, and (iv) the rest. The integral over (iii) is easily seen (by a change of variables) to be bounded by $\con\,s^9/\La^3$, which vanishes in the limit for $\eta<1/3$. Finally, we consider the region (iv). First, we note that
\begin{multline*}
 (\la_1v_1+\la_2v_2)^2=2\la_1^2+2\la_2^2-(\la_1v_1-\la_2v_2)^2\\[1ex] \leq 4\La_2^2 +|\la_1\vec{v}_1-\la_2\vec{v}_2|^2
 \leq 4\La_2^2 +4\tau^2s^2/(1-\beta)^2\,,
\end{multline*}
and for sufficiently large $\La$ this is bounded by $4(\La_2+\tau s)^2$. The same holds for $\la_1v_3+\la_2v_4$. Therefore
\begin{equation*}
 \tfrac{1}{2}(\la_1v_1+\la_2v_2)=\xi_1u_1\,,\quad \tfrac{1}{2}(\la_1v_3+\la_2v_4)=\xi_2u_2\,,
\end{equation*}
where $u_1,u_2\in H_+$,  $\La_1\leq\xi_i\leq\La_2+\tau s=\La_1+2\tau s$, and $|\vec{u}_i|\leq \La_2\nu/\La_1\leq\tau_2\nu/\tau_1\equiv\nu'$ (for $\La\geq 1/m$). We put $\beta'=\nu'/\sqrt{{\nu'}^2+1}$ and $\al=\gamma/(1-\beta')$, with $\gamma\geq1$ to be further specified. We note that now
\begin{equation*}
 |\xi_1-\xi_2|\leq 2\tau s\leq 2\gamma\tau s\,, \quad |\xi_1\vec{u}_1-\xi_2\vec{u}_2|\geq 4\gamma\tau s/(1-\beta')\,.
\end{equation*}
Thus using again the relations of Paragraph~\ref{hrgeom}, with $\sigma=2\gamma\tau s$ and $\beta\to\beta'$, we find
\begin{equation}\label{xivec0}
 |\xi_1\vec{u}_1-\xi_2\vec{u}_2|-|\xi_1u_1^0-\xi_2u_2^0|\geq
 \tfrac{1}{2}(1-\beta')|\xi_1\vec{u}_1-\xi_2\vec{u}_2|
 \geq2\gamma\tau s\,.
\end{equation}
Moreover, with the use of relation \eqref{vec0} one obtains
\begin{equation*}
 |\la_1v_1-\la_2v_2|\leq bs\,,\quad |\la_1v_3-\la_2v_4|\leq bs\,,
\end{equation*}
where $b=\sqrt{8}\tau/(1-\beta)$. It is now visible that for large enough $\La$ (and consequently $s$), with $d=bs$ and $\gamma$ chosen large enough to satisfy $2\gamma\tau/b=\gamma(1-\beta)/\sqrt{2}>c_1$, the premisses of Assumption \ref{cluster} are satisfied. Therefore, in this region
\begin{multline*}
 |K(\la_1v_1-\la_2v_2,\la_1v_3-\la_2v_4,\tfrac{1}{2}(\la_1v_1+\la_2v_2-\la_1v_3-\la_2v_4))|\\
 \leq\con\frac{s^M}{|\la_1\vec{v}_1+\la_2\vec{v}_2-\la_1\vec{v}_3-\la_2\vec{v}_4|^\vep}+\con\, s^{-N}
\end{multline*}
(use also the first inequality in \eqref{xivec0}). We change the variables $\vec{v}_1$, $\vec{v_2}$ and $\vec{v}_3$ to $\vec{w}_1=\la_1\vec{v}_1-\la_2\vec{v}_2$, $\vec{w}_2=\la_1\vec{v}_3-\la_2\vec{v}_4$ and $\vec{w}=\la_1\vec{v}_1+\la_2\vec{v}_2-\la_1\vec{v}_3-\la_2\vec{v}_4$ and note that $|\vec{w}_i|\leq 2\tau s/(1-\beta)$ and $8\gamma\tau s/(1-\beta')\leq|\vec{w}|\leq 4\La_2\nu$. Thus the integral over region (iv) is bounded by
\begin{equation*}
 \con\, s^6\La_1^{-3}\int\limits_{8\gamma\tau s/(1-\beta')}^{4\La_2\nu}\Big(\frac{s^M}{|\vec{w}|^\vep}+s^{-N}\Big)|\vec{w}|^2\,d|\vec{w}|
 \leq\con\,\big(s^{M+6}\La^{-\vep'}+s^{6-N}\big)\,,
\end{equation*}
where $\vep'=\min\{\vep,3\}$. This vanishes in the limit, if $N>6$ and \mbox{$\eta<\vep'/(M+6)$}.
\end{proof}

\paragraph{} 
The Fock structure of the products
\begin{equation}
 \lim_{\La\to\infty}(\Psi^\eta_{1\La}[f_1]\ldots\Psi^\eta_{k\La}[f_k]\W,\Psi_{(k+1)\La}^\eta[f_{k+1}]\ldots\Psi_{n\La}^\eta[f_{n}]\W)
\end{equation}
is now easily obtained by transferring the operators from the left to the adjoints on the right, commuting them to far right and using \eqref{two} (see \cite{dy05} for details of the technique).

\section{Conclusions}

We have introduced an asymptotic limiting of fields based on averaging over hyperboloids rather than constant time hyperplanes. If a class of fields satisfies a~rather slow spacelike decay condition of their (anti-)commutators, then their asymptotic behavior is naturally related to their spectral properties with respect to energy-momentum. In that case the asymptotic behavior admits a condition which generalizes the condition of the existence of a discrete mass hyperboloid in the energy-momentum spectrum in the vacuum representation. The resulting asymptotic fields transfer energy-momentum with sharp Lorentz square, interpreted as mass squared of a particle. With some stronger assumptions on the asymptotic limiting the asymptotic fields satisfy fermionic/bosonic statistics (but not the Fock structure).

The question whether the scheme will have relevance for realistic quantum electrodynamics is an open problem. However, a model proposed some time ago as an algebra of asymptotic fields in electrodynamics, in which Gauss' law is respected, fits into the scheme. It is an important problem for future research to find more general conditions for non-vanishing of the asymptotic limit fields as defined in the present paper.

The ideas at the base of these constructions were also put to a slightly different use to generalize the Haag-Ruelle scattering theory. It was shown that they allow some sharpening of results, while at the same time substantially relaxing assumptions on spacelike decay properties.

In the article only outgoing fields were considered, but incoming case strictly parallels these constructions.

\section*{Appendix}
\setcounter{section}{0}
\renewcommand{\thesection}{\Alph{section}}

\section{On regular wave packets}\label{regwave}

Here some properties of wave packets are discussed in a sharper form needed in this article, than usually considered.

Let $f$ be a smooth function on $H_+$, of compact support. Then for \mbox{$v\in H_+$}, $\rho>0$, there is
\begin{equation}\label{regwp}
 \int f(u)e^{i\rho v\cdot u}\,d\mu(u)
 =e^{i3\pi/4}\lp\frac{2\pi}{\rho}\rp^{3/2}e^{i\rho}\,
 \bigg(\sum_{k=0}^N\rho^{-k}L_kf(v)+O(\rho^{-N-1})\bigg)\,,
\end{equation}
where $L_0=\id$ and $L_k$ for $k\geq1$ are differential operators with smooth coefficient functions. The bound of the rest inside the parentheses, and of its $v$-derivatives, is uniform on each compact set of $v$'s. Moreover, each differentiation of the rest with respect to $\rho$ increases its decay rate by one power, with preserved uniformity. This follows from direct application of the stationary phase method, as presented in \cite{va89} (or, in a somewhat less explicit way, in \cite{rs79}). Note that the lhs of \eqref{regwp} is a~regular wave packet in the vector variables $\rho v$ covering $V_+$.

Choose now $f_j$, $j=0,\ldots,N$, with $f_0=f$, and other $f_j$ with similar regularity properties, and substitute in the above formula $f\to f_j$, \mbox{$N\to N-j$}. Combining the resulting formulae one finds
\begin{multline*}
 \sum_{j=0}^N \rho^{-j}\int f_j(u)e^{i\rho v\cdot u}\,d\mu(u)\\
 =e^{i3\pi/4}\lp\frac{2\pi}{\rho}\rp^{3/2}e^{i\rho}
 \bigg(\sum_{k=0}^N\rho^{-k}\sum_{j=0}^k L_{k-j}f_j(v)+O(\rho^{-N-1})\bigg)\,.
\end{multline*}
Putting now recursively $\dsp f_k=-\sum_{j=0}^{k-1}L_{k-j}f_j$ for $k=1,\ldots,N$ one obtains
\begin{equation*}
 \sum_{j=0}^N \rho^{-j}\int f_j(u)e^{i\rho v\cdot u}\,d\mu(u)\\
 =e^{i3\pi/4}\lp\frac{2\pi}{\rho}\rp^{3/2}e^{i\rho}
 \big(f(v)+O(\rho^{-N-1})\big)\,,
\end{equation*}
where the rest inside the parenthesis has the same properties as that in \eqref{regwp}.
\begin{pr}\label{exp}
Let $\wh{\chi}$ be a smooth real function on $M$ with compact support contained inside the future lightcone, and denote $\wh{\chi_j}(p)=\wh{\chi}(p)(p^2)^{3/4-j/2}$, $j=0,\ldots,N$.
Then, with standing assumptions and notation, for $\la>0$, there is
\begin{multline}\label{fexp}
 e^{-i3\pi/4}\sum_{j=0}^N \la^{-j}\Big(\frac{\la}{2\pi}\Big)^{3/2}\wh{\chi_j}(p)\int f_j(u)e^{i\la p\cdot u}\,d\mu(u)\\
 =e^{i\la\sqrt{p^2}}\wh{\chi}(p)f\big(p/\sqrt{p^2}\big)+\wh{R_\la}(p)\,,
\end{multline}
where $\wh{R_\la}$ is smooth, of compact support, and satisfies the bounds
\begin{equation}\label{Rh}
 |D^\al\wh{R_\la}(p)|\leq\con\, \la^{-N-1+|\al|}\,.
\end{equation}
The latter bounds imply
\begin{equation}\label{R}
 \int|R_\la(x)|dx\leq\con\,\la^{-N+3/2}\,.
\end{equation}
\end{pr}
\begin{proof}
All statements, except for the last estimate, follow directly from the preceding discussion, with $v=p/\sqrt{p^2}$ and $\rho=\la\sqrt{p^2}$. To show \eqref{R}, we note that for $k=0,1,2,\ldots$ the bounds \eqref{Rh} imply
\begin{equation*}
 \int|R_\la(x)|^2|x|^kdx \leq \con\, \la^{-2N-2+k}\,.
\end{equation*}
For even $k$ this follows directly by the Plancherel formula, and then for $k=2l+1$ by writing $|R_\la(x)|^2|x|^{2l+1}=|R_\la(x)||x|^l\times|R_\la(x)||x|^{l+1}$ and using the Schwartz inequality. Writing
\begin{equation*}
 |R_\la(x)|=|R_\la(x)|(|x|^5+1)^{1/2}\times(|x|^5+1)^{-1/2}
\end{equation*}
and using the Schwartz inequality we arrive at \eqref{R}.
\end{proof}

\section{Decay of correlations of commutators -- local case}\label{deccom}

Here we adapt an estimate due to Araki, Hepp and Ruelle \cite{ar62} to obtain the following result.
\begin{pr}\label{decay}
Let $\Psi_i$, $i=1,\ldots,4$, be strictly local (in the bosonic or fermionic sense) field operators, localized in the double cone $\C_R$. Then Assumption \ref{cluster} is satisfied for $d\geq R$, with $c_1=8$, $M=3$, $\vep=2$ and $c_3=0$.
\end{pr}

We note that another form of cluster property (due to Buchholz \cite{bu77}) was used in the Haag-Ruelle theory discussion by Dybalski (\cite{dy05}, Lemma 3.1). However, that result is not sharp enough for our purpose.

We begin by stating the original result (\cite{ar62}, formula (3.4)\footnote{There is a misprint in this formula in the original article: constants $C_0$ and $C_1$ should exchange their places, as is clear from the derivation (and for dimensional reasons).}) in the following form. We denote $\p_0B=\p_0B(x)_{|x=0}$.
\begin{thm}
Let $B_1$ $B_2$ be local operators, localized in the double cones $\C_{r_1}$, $\C_{r_2}$, respectively. Then for $|\vec{y}|\geq|y^0|+r$, $r=r_1+r_2$, the following estimate holds
\begin{equation}
 |(\W,B_1E_\W^\bot U(y)B_2\W)|\leq\frac{\con\ r^3}{(|\vec{y}|-r)^2-|y^0|^2}\Big[C_1+C_0\frac{|y^0|}{(|\vec{y}|-r)^2-|y^0|^2}\Big]\,,
\end{equation}
where $\con$ is a universal constant, $C_1=\|\p_0B_1\W\|\|B^*_2\W\|+\|\p_0B_2\W\|\|B^*_1\W\|$ and $C_0=\|B_1\W\|\|B^*_2\W\|+\|B_2\W\|\|B^*_1\W\|$.
\end{thm}
\begin{col}
If in the above theorem $|\vec{y}|\geq|y^0|+2r$, then
\begin{equation}\label{ahrcol}
 |(\W,B_1E_\W^\bot U(y)B_2\W)|\leq \con\, C\frac{r^3}{(|\vec{y}|-|y^0|)^2}\,,
\end{equation}
where $C=C_1+C_0/2r$.
\end{col}
\begin{proof}
In this case there is  $|\vec{y}|-r-|y^0|\geq(|\vec{y}|-|y^0|)/2$, so $(|\vec{y}|-r)^2-|y^0|^2$ is bounded from below by $(|\vec{y}|-|y^0|)^2/4$ and also by $2r|y^0|$, which implies the result.
\end{proof}

\noindent
{\bf Remark}\ \
In fact in this case a sharper form of the estimate \eqref{ahrcol}, with $\con\,r^3/|y^2|$ on the rhs, is also valid. However, we deliberately use the weaker form for generalization in the Haag-Ruelle-type construction, Section \ref{hrtheory}.

\begin{proof}[Proof of Proposition \ref{decay}]
We use the formula \eqref{4psi}. It is a simple geometric fact that in the present case $B_{12}(x_1-x_2)$ and $B_{34}(x_3-x_4)$ are localized in $\C_{R'}$ and $\C_{R''}$ respectively, with $R'=R+(|x_1^0-x_2^0|+|\vec{x}_1-\vec{x}_2|)/2\leq 2d$ and $R''=R+(|x_3^0-x_4^0|+|\vec{x}_3-\vec{x}_4|)/2\leq 2d$ (for the bounds the assumptions of the Proposition were used). The use of Corollary with $r=4d$ gives now immediately the thesis.
\end{proof}

\end{document}